\documentclass[preprint]{elsarticle}
\usepackage{array,xspace,multirow,hhline,tikz,colortbl,tabularx,booktabs,fixltx2e,amsmath,amssymb,amsfonts,amsthm}
\usepackage{algorithm}
\usepackage{algorithmic}
\usepackage{verbatim,ifthen}
\usepackage{enumitem}
\usepackage{pifont}
\usepackage{ifthen}
\usepackage{calrsfs,mathrsfs}
\usepackage{bbding,pifont}
\usepackage{pgflibraryshapes}
\usetikzlibrary{trees}
\usetikzlibrary{positioning,chains,fit,shapes,calc}
\usepackage{tkz-graph}

\usepackage{eucal}

\usepackage{tikz}
\usetikzlibrary{arrows}
\usetikzlibrary{decorations.pathreplacing}

	\usepackage{varioref}

\definecolor{light-gray}{gray}{0.9}

\newtheorem{definition}{Definition}%

	\newcommand{\eg}{e.g.,\xspace}

	\newtheorem{lemma}{Lemma}%
		\newtheorem{remark}{Remark}%
	\newtheorem{theorem}{Theorem}%
	\newtheorem{observation}{Observation}%
	\newtheorem{example}{Example}
		\newcommand{\pref}{\ensuremath{\succsim}\xspace}
		\newcommand{\spref}{\ensuremath{\succ}}
		\newcommand{\indiff}{\ensuremath{\sim}}
		
			\newcommand{\nb}{\ensuremath{NB~}}

	\newcommand\eat[1]{}

	\usepackage{enumitem}
	\setenumerate[1]{label=\rm(\it{\roman{*}}\rm),ref=({\it\roman{*}}),leftmargin=*}
	\newlength{\wordlength}

\RequirePackage{ifthen,calc}
\newsavebox{\ffbox}\newlength{\ffboxlen}
\newcommand{\todo}[1]{%
	\vskip4mm
	{\sbox{\ffbox}{\textbf{\color{red} TODO:}\ \textit{{#1}}\ \textbf{\color{red} :ODOT}}
    \settowidth{\ffboxlen}{\usebox{\ffbox}}
		\addtolength{\ffboxlen}{-5mm}
    \ifthenelse{\ffboxlen>\linewidth}{%
	\noindent\marginpar{$>>>>$}\textbf{\color{red} TODO:}\ \textit{{#1}}\ \textbf{\color{red} :ODOT}\marginpar{$<<<<$}}{%
	  \noindent\marginpar{$>><<$}\textbf{\color{red} TODO:}\ \textit{{#1}}\ \textbf{\color{red} :ODOT}}}
	\vskip4mm
  }

\usepackage{enumitem}
\setenumerate[1]{label=\rm(\it{\roman{*}}\rm),ref=({\it\roman{*}}),leftmargin=*}

\newcommand{\nbh}[1][]{
	\ifthenelse{\equal{#1}{}}{\nu}{\nu(#1)}
}

\newcommand{\cstr}[1][]{
	\ifthenelse{\equal{#1}{}}{\mathscr S}{\cstr(#1)}
}

\newcommand{\choice}[1][]{
	\ifthenelse{\equal{#1}{}}{\mathit{C}}{\choice(#1)}
}

\usepackage{boxedminipage}
\usepackage{xspace}

		\tikzset{
			inner sep=0pt, outer sep=0pt, minimum size=0pt, thick,
			level/.style={sibling distance = (\columnwidth/16)*2^(4-#1)},
			winner/.style={minimum size=1.5em, circle, draw, fill=white, font={\footnotesize}},
			leaf/.style={inner sep=.15em, font={\footnotesize}},
			ball/.style={minimum size=.4em,circle,fill=black},
			beats/.style={thick,->,>=stealth',draw}
		}

\begin{document}

\title{Testing Top Monotonicity}
% \title{Random assignment of multiple objects to agents}
	%\tnotetext[t1]{Thanks!} 
	% \tnotetext[t2]{The second title footnote which is a longer than the first one and with an intention to fill
	% in up more than one line while formatting.}
		\author{Haris Aziz}%\corref{cor1}} 
	\ead{haris.aziz@nicta.com.au}
		\address{NICTA and UNSW Australia, Kensington 2033, Australia}
% 	\author{Haris Aziz}%\corref{cor1}} 
% \ead{haris.aziz@nicta.com.au}
% 	\address{NICTA and UNSW Australia, Kensington 2033, Australia}
				% \address{NICTA, 223 Anzac Parade, Sydney, NSW 2052, Australia} 
	%	\author{Markus Brill} \ead{brill@in.tum.de}
		%	\address{Institut f\"ur Informatik, Technische Universit\"at M\"unchen, 85748 Garching bei M\"unchen, Germany} 
	%\author{Paul Harrenstein} \ead{paul.harrenstein@cs.ox.uk}
		%	\address{Department of Computer Science, University of Oxford, Oxford OX1 3QD, UK} 

%	\cortext[cor1]{Corresponding author} 

%	\fntext[fn1]{Thanks!} 
%%

\begin{abstract}

	Top monotonicity is a relaxation of various well-known domain restrictions such as \emph{single-peaked} and \emph{single-crossing} for which 
	negative impossibility results are circumvented and for which
	the \emph{median-voter theorem} still holds. We examine the problem of testing top monotonicity and present a characterization of top monotonicity with respect to non-betweenness constraints.
We  then extend the definition of top monotonicity to partial orders and show that testing top monotonicity of partial orders is NP-complete.  
	%Finally, we show that the first condition of top monotonicity can be tested in polynomial time for weak orders.

% Top monotonicity is a relaxation of various well-known domain restrictions such as \emph{single-peaked} and \emph{single-crossing} for which 
% negative impossibility results are circumvented and for which
% the \emph{median-voter theorem} still holds. We examine the problem of testing top monotonicity and present a characterization of top monotonicity with respect to non-betweenness constraints.
% Secondly, we show that top monotonicity of dichotomous preference profiles can be tested in linear time.
% Thirdly, we  extend the definition of top monotonicity to partial orders and show that testing top monotonicity of partial orders is NP-complete.  
% Finally, we show that testing `near top monotonicity' with respect to voter or alternative deletion is NP-complete even for dichotomous preference profiles.
%The algorithm may also be helpful in conducting interesting experiments regarding the structure of preferences in real-life data.
\end{abstract}

\begin{keyword}
Social choice theory \sep 
domain restrictions \sep
top monotonicity \sep
single peaked\sep
single crossing \sep
computational complexity.\\
	\emph{JEL}: C63, C70, C71, and C78
\end{keyword}

\maketitle

%\today

\section{Introduction}

The standard social choice setting is one in which a set of agents $N$ express preferences over a set of alternatives $A$ and the goal is the select an alternative based on the preferences. 
Social choice theory is replete with results concerning the impossibility of the voting rules simultaneously satisfying desirable axioms. 
Prominent among these results are \emph{Condorcet's Paradox}~\citep{Gehr06a}, \emph{Arrow's Theorem}~\citep{Arro63a}, and the \emph{Gibbard Sattherwaite Theorem}~\citep{Gibb73a}.
One approach to circumvent such results is identifying restricted domains in which impossibility results disappear. Examples of domain restrictions include preference profiles that are 
\emph{single-peaked; single-plateaued}~\citep{Moul80a} \emph{single crossing}~\citep{Robe77a}; or \emph{order restricted}~\citep{GaSm96a}. 
For example, the Condorcet paradox does not occur when preferences  are single-peaked. The \emph{median voter theorem} states that
for single-peaked preference profiles, Condorcet winner(s) exist and they  coincide with the median(s) of the voters' most preferred alternatives~\citep{Blac48a}.

Restricting domains to test robustness of impossibility results, and to identify settings that admit attractive preference aggregation rules, has been a fruitful area of research within social choice~\citep{BBM13a}.
Whereas economists have examined axiomatic implications of domain restrictions~\citep[see \eg][]{SePa69a,Inad69a,Ina64a,Roth90a,Moul80a}, computer scientists are starting to examine natural problems such as checking whether a preference profile satisfies a certain structural property.
On the computational side, \citet{BaTr86a} and \citet{ELO08a} showed that checking whether a profile is single-peaked is polynomial-time solvable. \citet{EFS12b} and \citet{BCW13a} then proved that it can be checked in polynomial time whether a profile is single-crossing. 
There is related work on checking whether \emph{other} computational problems such as manipulation become easy when the preference profile satisfies some structure~\citep[see \eg][]{BBHH10a}.
There has also been some work on identifying `almost' nicely structured profiles~\citep[see \eg][]{BCW13b,ELP13a}.
%However, we are interested in the more fundamental problem of \emph{checking} whether a preference profile satisfies a given structural property. 

Whereas the domain restrictions listed above have garnered a lot of attention, axiomatic results for these restrictions have been somewhat piecemeal. Recently, \citet{BaMo11a} proposed a new consistency condition called \emph{top monotonicity} that is a relaxation of all the domain restrictions listed above but which still is a sufficient condition for an extension of the median voter theorem to hold. Top monotonicity also allows for indifferences both among the maximally preferred alternatives and non-maximally preferred alternatives.\citet{BaMo11a} write that \emph{``Among other things, top monotonicity will stretch the extent to which one may accommodate indifferences and still obtain positive results regarding Condorcet winners in the case of the majority rule, or of voting equilibria, more generally.''} Top monotonicity is also induced in two natural  of tax rate determination~\citep[Appendix A, ][]{BaMo11a}.

Although the axiomatic aspects of top monotonicity have been studied~\citep{BBM13a,BaMo11a}, it is not clear how easy it is to test top monotonicity of a preference profile. \citet{BaMo11a} ask that \emph{``Is the satisfaction of our conditions easy to check? Top monotonicity may be easy to check for in some cases, and also easy to discard, in others.''}
In view of the generality of top monotonicity, similar computational investigations of restricted domains~\citep{BaTr86a,ELO08a,Knob10a,ABH13b,Lack13a}, and the questions raised by \citet{BaMo11a}, we study the problem of testing top monotonicity.

% \paragraph{Contributions}
\begin{itemize}
	\item We first present a reduction of testing top monotonicity to a previously studied problem of checking the  consistency of non-betweenness constraints~\citep{GuMa06a}. Since checking consistency of non-betweenness constraints is NP-complete, this reduction does not give us a polynomial-time algorithm for testing top monotonicity but it does help frame the problem in a way so it can be handled by solvers that deal with ordering constraints. 
% \item	For dichotomous preferences, we present a linear-time algorithm to test top monotonicity via an interesting connection with the \emph{consecutive ones property} of matrices.
\item 	 We then extend the definition of top monotonicity to partial orders in a natural way and show that testing top monotonicity of partial orders is NP-complete.  
\item Finally, we highlight that a simple profile with dichotomous preferences may fail to satisfy top monotonicity. 
%\item We present a linear-time algorithm to test the first condition of top monotonicity via an interesting connection with the \emph{consecutive ones property} of matrices.
% \item Finally, we consider `near top monotonicity' with respect to voter or alternative deletion and show that testing near top monotonicity is NP-complete even for dichotomous preference profiles.
\end{itemize}

% \begin{quote}
% \emph{What is the computational complexity of checking whether a profile satisfies top monotonicity?}
% \end{quote}
% 
% In fact \citet{BaMo11a} raise the same question in different words \emph{``Is the satisfaction of our conditions easy to check? Top monotonicity may be easy to check for in some cases, and also easy to discard, in others.''}
% This type of question has already received coverage for other restrictions as such as \emph{1-Euclidean}~\citep{Knob10a} and substitutability~\citep{ABH13b}. 
% We show that although the straightforward algorithm to test top monotonicity requires going through $|A|!$ orders, there exists a polynomial-time algorithm to test top monotonicity. The algorithm relies on our characterization of top monotonicity with respect to certain violation constraints.

% In another paper, \citet{Knob10a} presented a polynomial-time algorithm to check whether a profile is \emph{1-Euclidean}.
% Testing structural properties of preferences has also received attention in other settings such as stable matchings~\citep{ABH13b}.

% What is interesting that even in an economics paper, \citet{BaMo11a} wonder aloud: \emph{``Is the satisfaction of our conditions easy to check? Top monotonicity may be easy to check for in some cases, and also easy to discard, in others.''}
% Such questions can best be answered via the framework and tools of algorithms and complexity.

% \section{Preliminaries}

\section{Top Monotonicity}

Consider the social choice setting in which there is a set of agents $N=\{1,\ldots, n\}$, a set of alternatives $A=\{a_1,\ldots, a_m\}$ and a preference profile $\pref=(\pref_1,\ldots,\pref_n)$ such that each $\pref_i$ is a complete and transitive relation over $A$.
We write~$a \pref_i b$ to denote that agent~$i$ values alternative~$a$ at least as much as alternative~$b$ and use~$\spref_i$ for the strict part of~$\pref_i$, i.e.,~$a \spref_i b$ iff~$a \pref_i b$ but not~$b \pref_i a$. Finally, $\indiff_i$ denotes~$i$'s indifference relation, i.e., $a \indiff_i b$ iff both~$a \pref_i b$ and~$b \pref_i a$.

We define by $t_i(S)$ the set of maximal elements  $\pref_i$ on S.
Let $A(\pref)$ be the family of sets containing $A$ itself, and also all triples of distinct alternatives where each alternative is top on $A$ for some agent $k\in N$ according to $\pref$.
We are now in a position to present the definition of top monotonicity.

\begin{definition}[Top monotonic]\label{def:tm}
	A preference profile $\pref$ is top monotonic iff there exists a linear order $>$ over $A$ such that 
	
	\begin{enumerate}
		\item \label{item:case1} $t_i(A)$ is a finite union of closed intervals for all $i\in N$, and 
		\item \label{item:case2} For all $S\in A(\pref)$, for all $i,j\in N$, all $x\in t_i(S)$, all $y\in t_j(S)$, and any $z\in S$, the following holds:
			\end{enumerate}
			\begin{align*}
			&	[x>y>z \text{ or } z>y>x] \implies
			\end{align*}
			\begin{align}
%			&	[x>y>z \text{ or } z>y>x] \implies \\
			&	[y\pref_i z \text{ if } z\in t_i(S)\cup t_j(S) \text{ and } y\succ_i z \text{ if } z\notin t_i(S)\cup t_j(S)].\small \label{tm:implication}
			\end{align}
			We will say that $\pref$ satisfies top monotonicity with respect to order $>$.
\end{definition}

% \begin{example}[Top monotonicity]
% 	\end{example}

\begin{remark}
Top monotonicity does not preclude cycles in majority comparisons, but only guarantees that these do not occur at the top of the majority relation.
\end{remark}

\begin{remark}
	Top monotonicity on a set of alternatives is not necessarily inherited on its subsets.
\end{remark}

We examine the problem of testing top monotonicity.
It follows from Definition~\ref{def:tm}, that top monotonicity of a profile with respect to a particular order can be checked easily.

\begin{observation}
	For a given order over the alternatives, it can be checked in $O(|A|^3\cdot |N|^3)$ whether the preference profile satisfies  top monotonicity with respect to that order.
\end{observation}

The observation implies that an order can be a certificate that a profile satisfies top monotonicity. Hence, the problem of testing top monotonicity is in complexity class NP.
In view of the observation above, one can straightforwardly design an algorithm to test top monotonicity: go through each of the $|A|!$ orders and checks whether top monotonicity is satisfied with respect to one of the orders.

\begin{observation}
	It can be checked in time $O(|A|^3\cdot |N|^3\cdot|A|!)$ whether the preference profile satisfies  top monotonicity.% with respect to that order.
\end{observation}

A natural question is whether it is possible to test top monotonicity \emph{without} having to go through all the $|A|!$ orders. 
First, we make another observation.

\begin{lemma}\label{lemma:reverse}
	If a preference profile satisfies  top monotonicity with respect to an order, it also does so with respect to the reverse of the order.
	\end{lemma}
	\begin{proof}
Note that for any $S\in A(\pref)$, and for each $i,j\in N$, each $x\in t_i(S)$, $y\in t_j(S)$, and $z\in S$, the conditions of $(i)$ and $(ii)$ in the definition of top monotonicity are not affected if $>$ is reversed.		
	\end{proof}

	\section{Top monotonicity and non-betweenness violation constraints}

	In the rest of the section, we will heavily use the idea of \emph{\nb (non-betweenness)} violation constraints on the prospective order on $A$. We now introduce these constraints.
	
	\begin{definition}[Non-betweenness violation constraint]
	For $x,y,z\in A$, we call a $(y,\{x,z\})$ a \emph{\nb (non-betweenness)} violation constraint for an order $>$ over $A$ if neither $x>y>z$ nor $z>y>x$ can hold.  
		
	\end{definition}

Next, we show that the existence of an order on $A$ satisfying an \nb constraint set $C$ can be checked efficiently.% without having to go through all the orders.

\begin{lemma}\label{lemma:graph-algo}
	For an \nb constraint set $C$ on alternative set $A$, checking whether constraints in $C$ can be satisfied by some linear order is NP-complete.
\end{lemma}
\begin{proof}
	\citet{GuMa06a} proved that this particular problem is NP-complete.
\end{proof}

	\begin{lemma}\label{lemma:partial-linear}
		For an \nb constraint set $C$ on alternative set $A$, there exists a partial order that satisfies the constraints in $C$ if and only if there exists a linear order that satisfies the constraints in $C$
	\end{lemma}
\begin{proof}
	($\Rightarrow$) If a partial order satisfies the constraints in $C$, then the arcs in the graph representing the partial order already capture all the non-betweenness constraints. Any additional arcs do not negate these constraints. Hence, \emph{any} linear extension of the partial order also satisfies the constraints in $C$. 
	
		($\Leftarrow$) This direction is trivial since a linear order is also a partial order.
\end{proof}

		\subsection{Characterizing top monotonicity}

					We characterize top monotonicity when the set of alternatives is finite. 
					
					\begin{remark}
					The first condition of top monotociity is trivially satisfied when $A$ is finite. %In a finite domain, every set is a union of disjoint intervals (since singletons are closed intervals).
						\end{remark}
				
					In the next lemma, we show how \nb constraints can be used to capture top monotonicity.

% \begin{lemma}\label{lemma:nb-for-case1}
% 	For an order $>$, and $i\in N$, $t_i(A)$ is a closed interval iff for any $x,z\in t_i(A)$ and any $y\notin t_i(A)$, the \nb constraint  $(y,\{x,z\})$ holds.
% \end{lemma}
% \begin{proof}
% 	Case~\ref{item:case1} requires that for any $x,z\in t_i(A)$ and any $y\notin t_i(A)$, $y$ is not between $x$ and $z$. Hence $(y,\{x,z\})$ has to hold.
% \end{proof}

\begin{lemma}\label{lemma:nb-for-case2}
	For an order $>$, case~\ref{item:case2} in definition of top monotonicity is always satisfied iff
	for all $S\in A(\pref)$, for all $i,j\in N$, all $x\in t_i(S)$, all $y\in t_j(S)$, and any $z\in S$ such that $\neg[y\pref_i z \text{ if } z\in t_i(S)\cup t_j(S) \text{ and } y\succ_i z \text{ if } z\notin t_i(S)\cup t_j(S)]$, the \nb constraint  $(y,\{x,z\})$ holds.
\end{lemma}
\begin{proof}
	%The reasoning for using \nb constraints for case~\ref{item:case2} is as follows.
		For each $S\in A(\pref)$, and for each $i,j\in N$, each $x\in t_i(S)$, $y\in t_j(S)$, and $z\in S$, we check which of the following relative orders do not lead to a lack of top monotonicity.
	The following relative orders over $x,y,z$ trivially satisfy the top monotonicity condition~\ref{item:case2} for set $S$ and voter $i$ because they do not feature in the antecedent of  implication~\eqref{tm:implication}.

		\begin{enumerate}
			\centering
			\item $y>x>z$
			\item $y>z>x$
			\item $x>z>y$
			\item $z>y>x$
		\end{enumerate}

		The following relative orders in which $y$ is in between $x$ and $z$ are only in contention if the consequent of implication~\eqref{tm:implication} holds.
		\begin{enumerate}
			\centering
			\item $x>y>z$
			\item $z>y>x$
		\end{enumerate}

		If the consequent holds, then any of the six relative orders are possible and there is no constraint for this 3-set of alternatives. If the consequent does not hold, then we can safely say that both $x$ and $z$ are on one side of $y$ i.e., $y$ is cannot be between $x$ and $z$. Thus in this case we have a constraint on the ordering over $A$ that $y$ cannot between $x$ and $z$. Therefore for each 3-set of alternatives, either there is no constraint or there is a constraint $(y,\{x,z\})$: that $y$ cannot be between $x$ and $z$. Note that $(y,\{x,z\})$ is a required constraint if even for one $i$, $j$ and $S$, the consequent of implication~\eqref{tm:implication} does not hold. 
\end{proof}

Next, we show that an algorithm to check whether a linear order satisfies \nb constraints can be used to test top monotonicity. Unfortunately checking whether a linear order satisfies \nb constraints is NP-hard in general.

\begin{theorem}
Testing top monotonicity polynomial-time reduces to checking whether there exists a linear order that satisfies \nb constraints.
\end{theorem}
\begin{proof}
	In the algorithm, we first construct a set $C$ of \nb constraints that an order must satisfy. 
	%Firstly we put in $C$, all \nb constraints pertaining to Lemma~\ref{lemma:nb-for-case1}. 
We put in $C$, \nb constraints pertaining to case~\ref{item:case2} of the definition of top monotonicity as follows. For each $S\in A(\pref)$, and for each $i,j\in N$, each $x\in t_i(S)$, $y\in t_j(S)$, and $z\in S$, we check whether the consequent of the implication \eqref{tm:implication} holds. If it does not hold then we add $(y,\{x,z\})$ to $C$. Note that $|C|\leq |A|^3$ and $C$ can be constructed in time $O(|N|^2\cdot|A|^3\cdot|A|^3)$. 
Now that we have constructed the constraint set $C$, we can use an algorithm to check whether \nb can be satisfied by a linear order to test whether $\pref$ is top monotonic.
\end{proof}

\section{Complexity of testing top monotonicity}

It is not straightforward to use the connection with the problem of checking satisfiability of \nb constraints to show that testing top monotonicity is NP-complete. We first show that if we consider partial orders instead of weak orders, then testing top monotonicity can be proved to be NP-complete.

Top monotonicity as defined by \citet{BBM13a} applies to preferences that are weak orders. In this section, we first propose a plausible extension of top monotonicity for partial orders.

\begin{definition}[Top monotonicity of partial orders]
A preference profile consisting of partial orders satisfies top monotonicity if it possible to extend the partial orders to weak orders in such a way so that the resulting profile of weak orders satisfies top monotonicity.
\end{definition}

We prove that testing top monotonicity of partial orders is NP-complete.

\begin{theorem}
	Testing top monotonicity of partial orders is NP-complete.
\end{theorem}
\begin{proof}
	We present a reduction from checking whether \nb constraints can be satisfied or not. 
	For each \nb constraint $\alpha_i=(b_i,\{a_i,c_i\})$, we construct a voter $i_1$ who has the partial order 
	\begin{itemize}
		\item 	$a_{i}\mathrel{\succ_{i_1}} c_i \mathrel{\succ_{i_1}} b_i$
		\item $a_i \mathrel{\succ_{i_1}} x$ for all $x\in A\setminus \{a_i\}$.
	\end{itemize}

We also construct a voter $i_2$ who prefers $c_i$ over all other alternatives, and voter $i_3$ who prefers $b_i$ over all other alternatives. We will refer to the partial order profile as $\pref$.

We claim that the set of \nb constraints is satisfiable if and only if the resulting preference profile is top monotonic. 
	One direction is easy: If the \nb constraints are not satisfiable then the second condition of top monotonicity cannot be satisfied for all the voters. Hence the preference profile is not top monotonic.

		Now assume that the \nb constraints are satisfiable. Then there exists an ordering $>$ over the alternatives $A$ for which the \nb constraints are satisfied. We use $>$ to extend $\pref$ the partial order profile of the voters to single-peaked total order profile $\pref'$. By Theorem 1~\citep{BBM13a}, we know that if a preference profile is single-peaked, then it is top monotonic. Hence by showing that $\pref'$ is single-peaked, we 
will show that $\pref$ is top monotonic.
 For each voter $i_1$, his preferences are $a_i\mathrel{\succ_{i_1}}c_i \mathrel{\succ_{i_1}} b_i$.  We extend his preferences to a total order $\pref_{i_1}'$ as follows. Firstly, we consider $a_i$ as the peak in the order $>$. We identify whether $c_i$ is left or right of $a_i$ in $>$. Whichever direction $c_i$ is in, we start putting alternatives in that direction in the preference list of $\pref_{i_1}'$ starting from alternatives nearest to $a_i$ and then moving away. If there are alternatives on the other side of $a_i$, then we also append them to the preference list of $\pref_{i_1}'$ by first starting with alternatives nearest to $a_i$ and then moving away. 
For each voter $i_2$, we take $c_i$ as the peak in $\pref_{i_2}'$ and then we put alternatives on one side with the alternatives nearest to $c_i$ put first and then we do the same with the alternatives on the other side of $c_i$.
For each voter $i_3$, we take $b_i$ as the peak in $\pref_{i_2}'$ and then we put alternatives on one side with the alternatives nearest to $b_i$ put first and then we do the same with the alternatives on the other side of $b_i$.
Now that we have constructed a total order for all the voters, we argue that it is single-peaked. 
For each voter, the peak is clearly defined. Moreover the preferences slope down from the peak on either possible side according to the order $>$. Hence $\pref'$ is single-peaked. We already know that single-peaked preferences are top monotonic~\citep[Theorem 1, ][]{BBM13a}. Hence there exists a completion of partial order profile $\pref$ that is top monotonic.
\end{proof}

Dichotomous preferences are one of the most natural restrictions on preferences in which voter places the alternatives in one of two equivalence classes~\citep[see \eg][]{BoMo04a,BMS05a,BrFi07c}. 
One may erroneously conclude that for dichotomous preferences, testing the feasibility of first condition of top monotonicity is enough to test top monotonicity. However, we show in the next example that whereas the first condition of top monotonicity may be trivially feasible, the preference profile may not be top monotonic.

\begin{example}
	Consider the profile:
\begin{align*}
	1: \{x\},\{y,z\}\\
		2: \{y\},\{x,z\}\\
			3: \{z\},\{x,y\}
\end{align*}

The profile trivially satisfies the first condition of top monotonicity.
% because each voter has a uniquely maximal alternative.
Since the profile is symmetric, consider any order $>:x>y>z$
Note that for triple $\{x,y,z\}\in A(\pref)$. Furthermore, $z\pref_1 y $, $z\notin t_1(\{x,y,z\})\cup t_2(\{x,y,z\})$. Hence the profile does not satisfy top monotonicity with respect to $>$ and in fact any order.

\end{example}

\section{Conclusion}

We initiated the study of testing top monotonicity of preference profiles. Although we were unable to settle the complexity of the problem for weak orders, we show that the problem reduces the checking whether a set of non-betweenness constraints can be satisfied or not.
 We then show that testing top monotonicity of partial orders is NP-complete. Finally, we highlight that a simple profile with dichotomous preferences may fail to satisfy top monotonicity. 

\section*{Acknowledgments}
% \paragraph{Acknowledgment}
NICTA is funded by the Australian Government through the Department of Communications and the Australian Research Council through the ICT Centre of Excellence Program. The author thanks Bernardo Moreno Jimenez for some helpful pointers.

%~\citep[see \eg][]{BCW13b,ELP13a}.
  % 
  %
 % \bibliography{../../pamas/abb,../../pamas/pamas,../../pamas/brandt,../../pamas/aziz}

% \appendix
% 
% \section{Notes}
% Testing \emph{possible} single-peakedness is NP-hard for partial orders~\citep{Lack13a}.
% What about complete relations which may admit indifferences?
% 
% 
% \section{Hardness}
% 
% 
% Reduction from betweenness to non-betweenness.
% 
% Non-betweenness (b,a,c), (a,c,b) for each (a,b,c) betweenness constraint.
% 
% 
% 
% 
% 
% 
% 
% (a,b,c)
% 
% We want that b should be between a and c. Agent likes b.
% 
% (a,b,c)
% 
% agent 1 likes a the most out of a,b,c.
% agent 2 likes c the most of a,b,c.
% agent 1 
% 
% 
% 
% For all $S\in A(\pref)$, for all $i,j\in N$, all $x\in t_i(S)$, all $y\in t_j(S)$, and any $z\in S$, the following holds:
% 	\begin{align*}
% 	&	[x>y>z \text{ or } z>y>x] \implies
% 	\end{align*}
% 	\begin{align}
% %			&	[x>y>z \text{ or } z>y>x] \implies \\
% 	&	[y\pref_i z \text{ if } z\in t_i(S)\cup t_j(S) \text{ and } y\succ_i z \text{ if } z\notin t_i(S)\cup t_j(S)].\small \label{tm:implication}
% 	\end{align}
% 
% 
% 

\end{document}